\documentclass[12pt]{article}

\usepackage[margin=1in]{geometry}
\usepackage[affil-it]{authblk}
\usepackage{amsmath,amsfonts,amsthm,amssymb}
\makeatletter
\usepackage{natbib,graphicx}
\usepackage[colorlinks=true, allcolors=blue]{hyperref}
\usepackage{xcolor}
\usepackage{makecell} 
\usepackage{newtxtext}
\usepackage[subscriptcorrection]{newtxmath}
\usepackage[plain,noend]{algorithm2e}

\makeatletter
\renewcommand{\algocf@captiontext}[2]{#1\algocf@typo. \AlCapFnt{}#2} % text of caption
\def\@algocf@capt@plain{top}
\renewcommand{\algocf@makecaption}[2]{%
  \addtolength{\hsize}{\algomargin}%
  \sbox\@tempboxa{\algocf@captiontext{#1}{#2}}%
  \ifdim\wd\@tempboxa >\hsize%     % if caption is longer than a line
    \hskip .5\algomargin%
    \parbox[t]{\hsize}{\algocf@captiontext{#1}{#2}}% then caption is not centered
  \else%
    \global\@minipagefalse%
    \hbox to\hsize{\box\@tempboxa}% else caption is centered
  \fi%
  \addtolength{\hsize}{-\algomargin}%
}
\makeatother

\newcommand{\bY}{ {\bf Y} }
\newcommand{\sX}{ { \mathbb{ X}} }
\newcommand{\by}{ {\bf y} }
\newcommand{\bz}{ {\bf z} }
\newcommand{\btheta}{ {\boldsymbol \theta} }
\DeclareMathOperator*{\essinf}{ess\,inf}
\DeclareMathOperator*{\esssup}{ess\,sup}
\newcommand{\R}{\mathbb{R}}
\newcommand{\I}{\mathbb{I}}
\newtheorem{assumption}{Assumption}
\newtheorem{example}{Example}
\newtheorem{remark}{Remark}
\newtheorem{lemma}{Lemma}
\newtheorem{theorem}{Theorem}
\newtheorem{proposition}{Proposition}
\newtheorem{corollary}{Corollary}

\title{On the fundamental limitations of multiproposal Markov chain Monte Carlo algorithms}
\date{\vspace{-5ex}}
\author{Francesco Pozza\thanks{Bocconi Institute for Data Science and Analytics, Bocconi University, Milan, Italy. \texttt{francesco.pozza2@unibocconi.it}}\;  and Giacomo Zanella\thanks{
Department of Decision Sciences and Bocconi Institute for Data Science and Analytics, Bocconi University, Milan, Italy. \texttt{giacomo.zanella@unibocconi.it}\\
\emph{Both authors acknowledge support from the European Union (ERC), through the Starting grant `PrSc-HDBayLe', project number 101076564.}}
}
\begin{document}

\maketitle

\begin{abstract}
We study multiproposal Markov chain Monte Carlo algorithms, such as Multiple-try or generalised Metropolis-Hastings schemes, which have recently received renewed attention due to their amenability to parallel computing.
First, we prove that no multiproposal scheme can speed-up convergence relative to the corresponding single proposal scheme by more than a factor of $K$, where $K$ denotes the number of proposals at each iteration.
This result applies to arbitrary target distributions and it implies that serial multiproposal implementations are always less efficient than single proposal ones. 
Secondly, we consider log-concave distributions over Euclidean spaces, 
proving that, in this case, the speed-up is at most logarithmic in $K$, which implies that even parallel multiproposal implementations are fundamentally limited in the computational gain they can offer.
Crucially, our results apply to arbitrary multiproposal schemes and purely rely on the two-step structure of the associated kernels (i.e.\ first generate $K$ candidate points, then select one among those). Our theoretical findings are validated through numerical simulations. 
\vspace{9pt}
\end{abstract}

\section{Introduction}
Multiproposal Markov chain Monte Carlo (MP-MCMC) algorithms \citep{frenkel1996understanding,liu2000multiple,tjelmeland2004using} have been the subject of considerable methodological and applied interest over the last two decades. Compared to classical single proposal schemes \citep{hastings1970monte}, MP-MCMC schemes consider multiple candidate states at each Markov chain iteration in order to speed up convergence to stationarity.
While this improvement in speed comes 
with higher computational cost per iteration, the computations associated to different proposed states can be performed in parallel, thus offering a promising way to deploy parallel computing to speed up MCMC convergence \citep{calderhead2014general,holbrook2023generating,gagnon2023improving}.

Despite their popularity, however, fundamental properties of MP-MCMC methods, such as the effective gain obtained by increasing the number of proposed states, are still not well understood. Indeed, the complications related to the multiproposal structure has mostly limited the availability of rigorous results to simplified settings (see e.g.\ \citet{yang2023convergence} and references therein). Moreover, unlike for the single proposal case, in the multiproposal context there is not a single selection rule that is uniformly optimal \citep{tjelmeland2004using}.
This makes it harder to disentangle features that are due to specific design choices \citep{craiu2007acceleration,pandolfi2010generalization,bedard2012scaling,martino2017issues,luo2019multiple,fontaine2022adaptive,gagnon2023improving} as opposed to ones that are fundamental to the MP-MCMC framework.

In this paper we provide rigorous theoretical results that apply to arbitrary MP-MCMC schemes, regardless of the specific algorithmic variant. In this way, we contribute to understand the intrinsic and fundamental properties and limitations of the MP-MCMC framework.
Specifically, we quantify the best possible speed-up of MP-MCMC schemes, relative to single proposal ones, as a function of the number $K \in \{1,2,\dots\}$ of proposed states (which, in a parallel computing context, corresponds to the number of parallel workers). Our first result is a Peskun ordering comparison (Theorem \ref{theorem:peskun}), which implies that the multiproposal speed-up can never be higher than $K$. The result is surprisingly simple and general, both in terms of target distribution, state space and algorithmic variant. This result implies that \emph{serial} implementations of MP-MCMC cannot improve over single proposal schemes in terms of resulting efficiency. However, it does not clarify if and when \emph{parallel} implementations of MP-MCMC are effective (see e.g.\ discussion after Corollary \ref{coroll:peskun_gap}). In Sections \ref{sec:neg2}-\ref{sec:neg3}, we thus make more concrete assumptions in order to give sharper results. 
In particular, we focus on log-concave target distributions defined on Euclidean spaces, showing that the resulting multiproposal speed-up (measured by the increase in the spectral gap of the associated Markov chain) is at most poly-logarithmic in $K$ (Theorem \ref{thm:GRW}). This result suggests that, at least for the framework we analyse, MP-MCMC schemes are fundamentally limited in the speed-up they can provide.
Notably, our results apply jointly to all MP-MCMC schemes, regardless of the specific rule used to select points (e.g.\ multiple-try with or without locally-balanced weights, generalised Metropolis-Hastings, antithetic proposals, etc), thus suggesting that better ways to speed up MCMC with parallel computing should be sought outside of the multiproposal framework. 
The results are illustrated by a simulation study on high-dimensional logistic regression (Section \ref{sec:sim}), which confirms the theoretical findings. All proofs are provided in the Appendix.

\section{Transition kernels based on multiple proposals} \label{sec:review}
Given a probability measure of interest $\pi$ on a sample space $\sX$, we consider MCMC schemes that generate $\pi$-reversible Markov chains according to the following structure: at each iteration, $K\geq 1$ candidate points are generated, subsequently either one of them is accepted as a new value or the chain remains at the current position. 
This general structure is described in Algorithm \ref{alg:meta}.
The algorithm depends on two key parameters: a Markov transition kernel $Q$ from $\sX$ to $\sX^K$, which describes the mechanism used to propose the $K$ points $ \by_{1:K}=( y_1,\dots, y_K)$, and a function $h$ from $\sX\times \sX^{K}$ to $S_{K}=\{ (s_1,\dots,s_K)\in[0,1]^{K}\,:\,\sum_{i = 1}^K s_i \leq 1\}$, where $h_i( x,  \by_{1:K})$ indicates the probability of selecting the $i$-th proposed point as a next state.
Crucially, we make no assumption on the specific form of $h$, which could be either given analytically or implicitly defined through some algorithmic procedure.
\begin{algorithm}
\caption{General multiproposal MCMC algorithm}\label{alg:meta}
\KwIn{$X_0 \in \sX$, $Q$ Markov transition kernel from $\sX$ to $\sX^{K}$, function $h:\sX^{1+K}\to S_K$}
\For{$t = 0,1,\dots$ 
}{
Given $X_t= x$, sample $ \by_{1:K}=( y_1,\dots, y_K) \sim Q( x,\cdot)$\\
Select $X_{t+1}$ from $( x, y_1,\dots, y_K)$ with probabilities $(1-\sum_{i=1}^K h_i( x,  \by_{1:K}),h_1( x,  \by_{1:K}), \dots,h_K( x,  \by_{1:K}))$\\
}
\end{algorithm}

Equivalently, 
denoting the indicator function of $A\subseteq \sX$ with $\mathbf{1}_{A }(\cdot)$, 
we consider Markov transition kernels of the form
\begin{align}  \label{kernel:meta:1}
P^{(K)}( x,A) 
&=
\sum_{i = 1}^K \int_{\sX^{K}} \mathbf{1}_{A}( y_i)h_i(  x,  \by_{1:K})Q(  x, d \by_{1:K}),
&A \subseteq \sX,\,x \notin A 
,
\end{align}
and $P^{(K)}( x, \{ x\}) = 1- P^{(K)}( x,\sX \backslash\{ x\})$.
The equivalence between Algorithm \ref{alg:meta} and the kernel representation in \eqref{kernel:meta:1} is shown in Proposition \ref{prop:equi:kern:MA} in the appendix.

When $K=1$, Algorithm \ref{alg:meta} reduces to classical accept-reject kernels \citep{tierney1998note}, and the class of functions $h$ that ensure $\pi$-reversibility of $P^{(K)}$ is well-understood and studied.
In particular, regardless of the choice of $Q$, it is well-known that the Metropolis-Hastings (MH) function, $h( x, y)=\min\{1,\pi(d y)Q( y,d x)/(\pi(d x)Q( x,d y))\}$, is optimal in the Peskun sense \citep{peskun1973optimum}.

When $K>1$, instead, designing practical choices of $h$ that ensure $\pi$-reversibility is more challenging, and there is no choice of $h$ uniformly dominating all others \citep{tjelmeland2004using}.
For this reason, various choices of $h$ have been proposed in the literature, potentially in combination with specific choices of $Q$, resulting in several different MP-MCMC schemes. 
Examples \ref{ex:MTM} and \ref{ex:GMH} below describe two popular options, which are both special cases of Algorithm \ref{alg:meta}.

\begin{example}[Multiple-try Metropolis]\label{ex:MTM}
Let $Q_i( x,\cdot)$ and $\bar{Q}_i^{ y_i}( x, \cdot)$
denote, respectively, the marginal distribution of $Y_i$ and the conditional distribution of $\bY_{-i}=(Y_j)_{j\neq i}$ given $Y_i =  y_i$, under $\bY_{1:K}\sim Q( x,\cdot)$.
Given the current and proposed points, $ x \in \sX$ and $ \by_{1:K}\in\sX^{K}$, Multiple-try Metropolis (MTM) methods \citep{frenkel1996understanding,liu2000multiple} select the new state $X_{t+1}$ according to the following procedure:\vspace{-2mm}
\begin{enumerate}
    \item[(i)] Sample $i$ from $\{1,\dots,K\}$ with probabilities proportional to $\{w_i( x,  y_i)\}_{i=1}^K$, where $w_i:\sX^2\to\R^+$ for $i=1,\dots,K$ are arbitrarily chosen weight functions,
    \item[(ii)] Sample $ y'_{-i} \sim \bar{Q}^{ x}_i( y_i, \cdot)$ and set $X_{t+1} =  y_i$ with probability
   $$
   \min \Big( \, 1\, , \,  \frac{
   \pi(d y_i)Q_i( y_i,d x) w_i( y_i,  x)/\big\lbrace w_i( y_i, x)+\sum_{j \neq i} w_j( y_i, y'_j)\big\rbrace 
   }{
   \pi(d x)Q_i( x,d y_i) w_i( x,  y_i)/\big\lbrace \sum_{j \in \{1, \dots, K\}}  w_j( x, y_j) \big\rbrace   } \Big)\,,$$
  otherwise set $X_{t+1} =  x$\,.
\end{enumerate}
\end{example}

\begin{example}[Tjelmeland's proposal]\label{ex:GMH}
\citet{tjelmeland2004using} proposes defining $h_i$ as
\begin{equation} \label{acc:pr:tj}
h_{i}( x, \by_{1:K}) = \frac{\pi( y_{i})q( y_{i},( x, \by_{-i})) 
}{
\pi( x)q( x, \by_{1:K})
+
\sum_{j =1}^{K} 
\pi( y_{j})q( y_{j},( x, y_{-j}))
}\,,
\end{equation}
where $\pi$ and $q(x,\cdot)$ are the density of, respectively, $\pi$ and $Q(x,\cdot)$, with respect to some dominating measures $\mu$ and its $K$-th power $\mu^K$, 
and
$
( x, \by_{-i})=( y_1,\dots, y_{i-1}, x, y_{i+1},\dots, y_{K})
$.
The expression in \eqref{acc:pr:tj} makes the kernel $P^{(K)}$ $\pi$-reversible for any choice of $Q$. 
However, 
in order to avoid $h_{i}( x, \by_{1:K})$ being too small, 
it is common to choose $Q$ so that
$q( y_{i},( x, \by_{-i}))=q( x, \by_{1:K})$ for all $( x, \by_{1:K})\in\sX^{1+K}$, in which case \eqref{acc:pr:tj}  simplifies to
$h_{i}( x, \by_{1:K}) = 
\pi( y_{i})/(\pi( x)
+
\sum_{j =1}^{K} 
\pi( y_{j})
)$.
For example, when $\sX=\R^d$, one option is to
generate $ \by_{1:K}\sim Q( x,\cdot)$ by first sampling $z\sim N( x,(\sigma^2/2)\I_d)$ and then $ y_i\sim N(z,(\sigma^2/2)\I_d)$ independently for $i=1,\dots,K$, where $\I_d$ denotes the $d\times d$ identity matrix and $N(\cdot,\cdot)$ denotes multivariate Gaussian distributions, so that $q( y_{i},( x, \by_{-i}))=q( x, \by_{1:K})$ is satisfied.
\end{example}

Many variants of MTM have been proposed in the literature, such as versions with antithetic proposals  \citep{craiu2007acceleration} and locally-balanced weights \citep{gagnon2023improving}. 
Similarly, Example \ref{ex:GMH} is a special case of a broad class of MP-MCMC schemes called Generalised Metropolis-Hastings (GMH) \citep{tjelmeland2004using, calderhead2014general, holbrook2023generating}. We refer to Section \ref{sec:review_MP_MCMC} of the appendix for a more detailed review of MP-MCMC methods.

\subsection{Multiproposal MCMC and parallel computing}\label{sec:par_comp}
In many contexts, the leading computational cost for each iteration of MP-MCMC methods lies in the multiple target evaluations, namely evaluating $\{\pi( y_i)\}_{i=1}^K$ for Example \ref{ex:GMH} and 
$\{\pi( y_i),\pi( y'_i)\}_{i=1}^{K}$
for Example \ref{ex:MTM}.
In those situations, each MP-MCMC iteration costs (roughly) $K$ times more than one iteration of standard MH, since the latter requires only a single evaluation of $\pi$ at each iteration. 
On the other hand, the multiproposal structure allows to perform these $K$ (or $2K-1$) likelihood computations in parallel, so that the effective runtime required to sample from $P^{(K)}$ can be much less than $K$ times the one of MH (potentially even being of the same order).
Thus, since $P^{(K)}$ converges faster to stationarity as $K$ increases, multiproposal schemes offer a way to exploit parallel computing to speed-up convergence of MCMC. Relative to the classical use of parallel computing in MCMC, where multiple independent chains are run in parallel \citep{rosenthal2000parallel}, this approach has the potential to reduce the time required to reach stationarity (i.e.\ the so-called called burn-in or warm-up phase), thus being particularly appealing in contexts where $\pi$ is expensive to evaluate and/or the available runtime (i.e.\ maximum number of possible iterations)
is limited. See e.g.\ \citet{calderhead2014general,gagnon2023improving,holbrook2023generating} for more discussion. 
However, in order to understand if and when the resulting trade-off is convenient, one needs to quantify the speed-up of $P^{(K)}$ as a function of $K$, which is what we focus on in the next section.

Note that the above considerations apply also to cases where $Q$ incorporates gradient information.
For example, if $\sX=\R^d$ and $Q(x,d\by_{1:K})=\prod_{i=1}^KQ_i(x,dy_{i})$ with $Q_i(x,\cdot)=N(x+\sigma^2/2\nabla\log\pi(x),\sigma^2\I_d)$ being a Langevin proposal, then each MP-MCMC iteration requires the computation of both $\pi$ and $\nabla\log\pi$ at all $y_i$ for $i\in\{1,\dots,K\}$, thus again being (at least) $K$ times more expensive than a single iteration of MH with proposal $Q_i$, unless parallel computing is used.

\section{Upper bounds on the spectral gap of multiproposal MCMC algorithms} \label{sec:teo:results}

In this section we consider MP-MCMC schemes of the form of Algorithm \ref{alg:meta}, providing upper bounds to their improvement as a function of $K$. 
Crucially, our results do not make any assumption on the form of $h$ nor on the dependence structure of $Q$, thus applying simultaneously to all the MP-MCMC variants described in Section \ref{sec:review} and the appendix.

\subsection{A general Peskun ordering result} \label{sec:neg1}

We first show that any $\pi$-reversible MP-MCMC scheme with proposal $Q$ is dominated, in the Peskun sense \citep{peskun1973optimum,tierney1998note}, by $K$ times the kernel of a single proposal MH algorithm with proposal $\tilde{Q}=K^{-1}\sum_{i=1}^KQ_i$, with $Q_i$ defined as in Example \ref{ex:MTM}. 
\begin{theorem}\label{theorem:peskun}
    Let $P^{(K)}$ be a $\pi$-reversible Markov transition kernel as in \eqref{kernel:meta:1}.
    Then
    \begin{equation} \label{mt:marginal:upper:bound:mh}
        P^{(K)}( x,A\backslash \{x\}) \leq K \tilde{P}( x,A\backslash \{x\}).
    \end{equation}
    for $\pi$-a.e. $ x$ and all measurable $A\subseteq \sX$, where $\tilde{P}$ is the (single proposal) MH kernel with proposal distribution $\tilde{Q}= K^{-1}\sum_{i=1}^K Q_i$ and target distribution $\pi$.
\end{theorem}
Peskun orderings provide a strong notion of dominance of one kernel over another and, in particular, Theorem \ref{theorem:peskun} can be interpreted as saying that $P^{(K)}$  can speed up convergence relative to $\tilde{P}$ by at most a factor of $K$. 
For example, \eqref{mt:marginal:upper:bound:mh} implies an analogous bound on the corresponding spectral gaps, see e.g.\ \citet[Lemma 32]{andrieu2018uniform} and \citet[Thm.2]{zanella20}.
Recall that the spectral gap of a $\pi$-reversible kernel $P$ is defined as 
\begin{equation*}
Gap(P) = \inf_{f\in L^2(\pi)} \frac{\int_{\sX^2} \lbrace f( y) - f( x) \rbrace^2 \pi(d x)P( x,d y) }{2\mathrm{Var}_{\pi}(f)},
\end{equation*}
where $L^2(\pi)$ denotes the collection of all 
$f \, : \, \sX \to \mathbb{R}$ such that $\mathrm{Var}_{\pi}(f)=\int_{\sX} f( x)^2\pi(d x) - (\int_{\sX}\pi(f) \pi(d x))^2  <\infty$. The spectral gap coincides with the difference between the first and second largest eigenvalue of the operator associated to $P$, and it is commonly used to measure the convergence properties of Markov kernels, with larger gaps related to faster covergence \citep[see e.g.,][Section 12]{levin2017markov}. 
\begin{corollary}\label{coroll:peskun_gap}
Under the same assumptions of Theorem \ref{theorem:peskun}, it holds 
$Gap(P^{(K)}) \leq K Gap(\tilde{P})$.
\end{corollary}
Since in most situations the computational cost per iteration of $P^{(K)}$ is at least $K$ times higher than the one of $\tilde{P}$
(for the reasons discussed in Section \ref{sec:par_comp}), Theorem \ref{theorem:peskun} implies that running a serial implementations of $P^{(K)}$ for $T$ iterations is less efficient than running $\tilde{P}$ for $K\times T$ iterations.
This confirms the intuition that, in a serial computing context, it is optimal to take $K=1$.

However, if one considers a parallel computing context, where $K$ represents the number of parallel workers, speeding up convergence by a factor of $K$ is not restrictive. On the contrary, one could argue that (ignoring additional over-heads due to communication costs) it represents an optimal usage of parallel resources.
In this sense, Theorem \ref{theorem:peskun} has strong (negative) implications for serial implementations of MP-MCMC, but not many direct implications for parallel ones.
In the next sections, however, we refine our analysis showing that in common situations the actual gain is much more limited.

\begin{remark}[Related results]
\citet{yang2023convergence} provide a full spectral analysis of $P^{(K)}$ in the specific case of the MTM with independent proposals, where $Q(x,\by_{1:K})$ does not depend on $x$, proving an upper bound on $Gap(P^{(K)})$ analogous to that given in Corollary \ref{coroll:peskun_gap}. 
In the context of MTM algorithms with locally-balanced weight functions applied to Bayesian model selection problems, \citet{chang2022rapidly} provide a lower bound on $Gap(P^{(K)})$ that grows linearly with $K$ for finite sample spaces $\sX$ whose dimensionality grows with $K$ at appropriate rates (see Theorem 1 therein). This implies that the upper bound in Corollary \ref{coroll:peskun_gap} can be tight in some cases.
\end{remark}

\subsection{Spectral gap behaviour for continuous target distributions} \label{sec:neg2}
We now consider the Euclidean case, $\sX = \mathbb{R}^d$, where more explicit bounds on $Gap(P^{(K)})$ can be derived. 
Our first result is the following.
\begin{theorem}\label{theo:1}
Let $P^{(K)}$ and 
$\tilde{P}$ as in Theorem \ref{theorem:peskun}, with $\sX \subseteq \mathbb{R}^d$.
Then 
\begin{equation} \label{theo1:result1}
\begin{aligned}
& Gap(P^{(K)})  \leq  \min  \Biggl(2 K \essinf_{ x \in \sX}  \tilde{P}( x,\sX \backslash \{ x\}), \, \inf_{\nu \in \mathbb{R}^d,\,\|\nu\|=1} 
\frac{ E(\max_{i = 1,\dots,K} 
(\nu^{\top}(Y_i -X))^2 
)}{2 \mathrm{Var}(\nu^{\top} X)} \Biggr),
\end{aligned}
\end{equation}
where 
$X \sim \pi$, $\bY_{1:K}|X\sim Q(X,\cdot)$, and 
$\essinf$ refers to the essential infimum with respect to $\pi$.
\end{theorem} 
The upper bound in \eqref{theo1:result1} is analogous to the ones available for standard MH schemes \citep[Thm.46]{andrieu2022explicit}. It has an intuitive structure which relates to the classical trade-off when choosing an appropriate step-size in accept-reject MCMC schemes: the first term gets small if the proposals are rejected too often (i.e.\ step-size too large) while the second gets small if the proposed jumps are too little (i.e.\ step-size too small). 
Crucially, the only aspects of \eqref{theo1:result1} that are related to the multiproposal context are (i) the appearance of the factor $K$ in the first term and (ii) the maximum across $K$ variables in the second one.
In common situations, the latter maximum scales logarithmically in $K$, resulting in upper bounds of the following form.
\begin{corollary}\label{coroll:mgf}
Under the same setting of Theorem \ref{theo:1}, assuming that $E(\exp(s (\nu^{\top}(Y_i -X))^2 )\leq M^{\nu}(s)<\infty$ for some $s \in \mathbb{R}^{+}$ and all $i \in \{1,\dots,K\}$, we have
\begin{equation} \label{theo1:result2}
\begin{aligned}
Gap(P^{(K)}) \leq \min  \Biggl(2 K \essinf_{ x \in \sX}  \tilde{P}( x,\sX \backslash \{ x\}),
\, 
\inf_{\nu \in \mathbb{R}^d,\,\|\nu\|=1}
\frac{1}{2s}\frac{ \log(K) + \log(M^{\nu}(s))}{ \mathrm{Var}(\nu^{\top} X)} \Biggr)\,.
\end{aligned}
\end{equation}
\end{corollary}
Keeping all other elements (such as $\pi$ and $\tilde{Q}$) fixed, the upper bound in \eqref{theo1:result2} scales as $\log(K)$ for $K\to\infty$.
This, however, does not allow us to deduce that, as $K$ grows, performances of MP-MCMC increase logarithmically with $K$, because in practice the optimal choice of $\tilde{Q}$ would change as $K$ grows (e.g.\ a user of an MP-MCMC scheme would increase the proposal step-size as
$K$ grows).
Instead, to understand the actual behaviour of $Gap(P^{(K)})$ as a function of $K$ (and potentially $d$), one needs to specify a class of proposal distributions for $Q$ and optimise the bound in \eqref{theo1:result2} over such class for every fixed $K$. 
In the next section, we do so for the case of random walk proposals, which are arguably the most commonly used ones in the MP-MCMC context.

\subsection{Analysis for random walk multiproposal schemes under log-concavity} \label{sec:neg3}
We make the following assumptions on $\pi$ and $Q$.

\begin{assumption} \label{cond:1}
$\sX \subseteq \mathbb{R}^d$ and $\pi$ admits density $\pi( x) = \exp(-U( x))$  with respect to Lebesgue, with $U( x)$ being $m$-convex, $L$-smooth and twice continuously differentiable. Equivalently, the Hessian of $U$, which we denote as $\mathbf{U}''(x)$, is such that  $\mathbf{U}''( x) - m\I_d$  and $L\I_d -\mathbf{U}''( x)$ are positive semi-definite matrices for every $ x\in \sX$.
\end{assumption}
\begin{assumption} \label{cond:2}
$Q_i( x,\cdot) = N( x,\sigma^2\I_d)$ for $x\in\sX$ and $i=1,\dots,K$, with $\sigma >0$.
\end{assumption}
Assumption \ref{cond:1} is a standard assumption in theoretical studies of MCMC methods, see e.g.\ \citet{chewi2024log} and references therein. 
Assumption \ref{cond:2} includes standard multiproposal implementations that rely on independent and identically distributed Gaussian proposals, as well more advanced MP-MCMC schemes with correlated proposals \citep{craiu2007acceleration}. Assumption \ref{cond:2} could be relaxed by allowing $\sigma$ to depend on $i$, at the cost of more involved statements: since the conclusions would be similar, we stick to Assumption \ref{cond:2} for simplicity. 

\begin{theorem}\label{thm:GRW}
Under Assumptions \ref{cond:1} and \ref{cond:2}, we have
\begin{equation} \label{thm:GRW:result}
\begin{aligned}
&Gap(P^{(K)})  \leq 
2\min  \Biggl( \frac{K}{(1 + m \sigma^2)^{d/2}},  
\sigma^2 L \left(\frac{1}{2}+\log(K)\right)\Biggr).
 \end{aligned}
\end{equation}
It follows that, for $c=4\exp(1.04)$ and all $d,K>2$, we have
\begin{align}\label{thm:GRW:result2}
\sup_{\sigma\in\R^+} Gap(P^{(K)})
&\leq
c 
\frac{L}{m}
\frac{(\log(K)+\log(d))^2}{d}\,.
\end{align} 
\end{theorem}

Theorem \ref{thm:GRW} provides finite-sample upper bounds for the spectral gap of every MP-MCMC under Assumptions \ref{cond:1}-\ref{cond:2}. The bound in \eqref{thm:GRW:result} depends on the proposal step-size $\sigma$, while the one in \eqref{thm:GRW:result2} considers optimally-tuned MP-MCMC schemes. As mentioned above, maximizing with respect to $\sigma$ is important, since in general the optimal value of $\sigma$ depends on $K$ in a non-trivial way (intuitively, larger $K$ allows the user to take larger steps).
The bound in \eqref{thm:GRW:result2} implies that, for $d,K \to \infty$ and $K$ growing polynomially with $d$, the spectral gap of optimally tuned MP-MCMC schemes is at most  of order $d^{-1}\log(K)^2$.
Since, under Assumptions \ref{cond:1} and \ref{cond:2}, \citet[Thm.1]{andrieu2022explicit} provides lower bounds of order $d^{-1}$ for MH with $K=1$, \eqref{thm:GRW:result2} implies that, in the context of Assumptions \ref{cond:1} and \ref{cond:2}, multiproposal schemes can provide at most a gain of order $\log(K)^2$ relative to the standard MH.

\section{A simulation study on high-dimensional logistic regression} \label{sec:sim}
We illustrate the above theoretical results through a simulation study, examining how $K$ affects performances.
We consider a $50$-dimensional Bayesian logistic regression model, with likelihood $z_i|\btheta \sim \hbox{Bernoulli}(1/(1+\exp(- \textbf{b}_i^\top \btheta))$ independently for $i\in\{ 1,\dots,n\}$ with $\textbf{b}_i\in \mathbb{R}^{d}$, and $\btheta \in \mathbb{R}^{d}$ having prior $\btheta \sim N(\mathbf{0}, (25/d)\I_d)$.
The target distribution is the posterior distribution $\pi(\btheta \mid \bz )$, with $n=d=50$ and observed data $\bz = (z_1,\dots,z_{n})$ generated by sampling from the likelihood with covariates and parameters randomly generated as $ \textbf{b}_i \sim N(\mathbf{0}, \I_{d})$ and $\btheta_0 \sim N(\mathbf{0}, (1/4) \I_{d})$.

We consider five different commonly used MP-MCMC schemes with random walk increments: MTM as in Example \ref{ex:MTM} with both independent \citep{liu2000multiple} and extremely antithetic \citep{craiu2007acceleration} zero-mean Gaussian increments, both with globally and locally balanced \citep{gagnon2023improving} weight functions, i.e.\ $w_i(\btheta,\btheta')=\pi(\btheta'|\bz)/\pi(\btheta|\bz)$ and $w_i(\btheta,\btheta')=(\pi(\btheta'|\bz)/\pi(\btheta|\bz))^{1/2}$; and finally \citet{tjelmeland2004using} star-shaped proposal described in Example \ref{ex:GMH}.
Note that the target $\pi(\btheta|\bz)$ satisfies Assumption \ref{cond:1} and, for each algorithm, the marginal distributions of $Q$ are Gaussian and identically distributed. Thus, the simulation setting satisfies the assumptions of Theorem \ref{thm:GRW}.

\begin{figure}
    \centering
    \includegraphics[width=14cm,height=6.5cm]{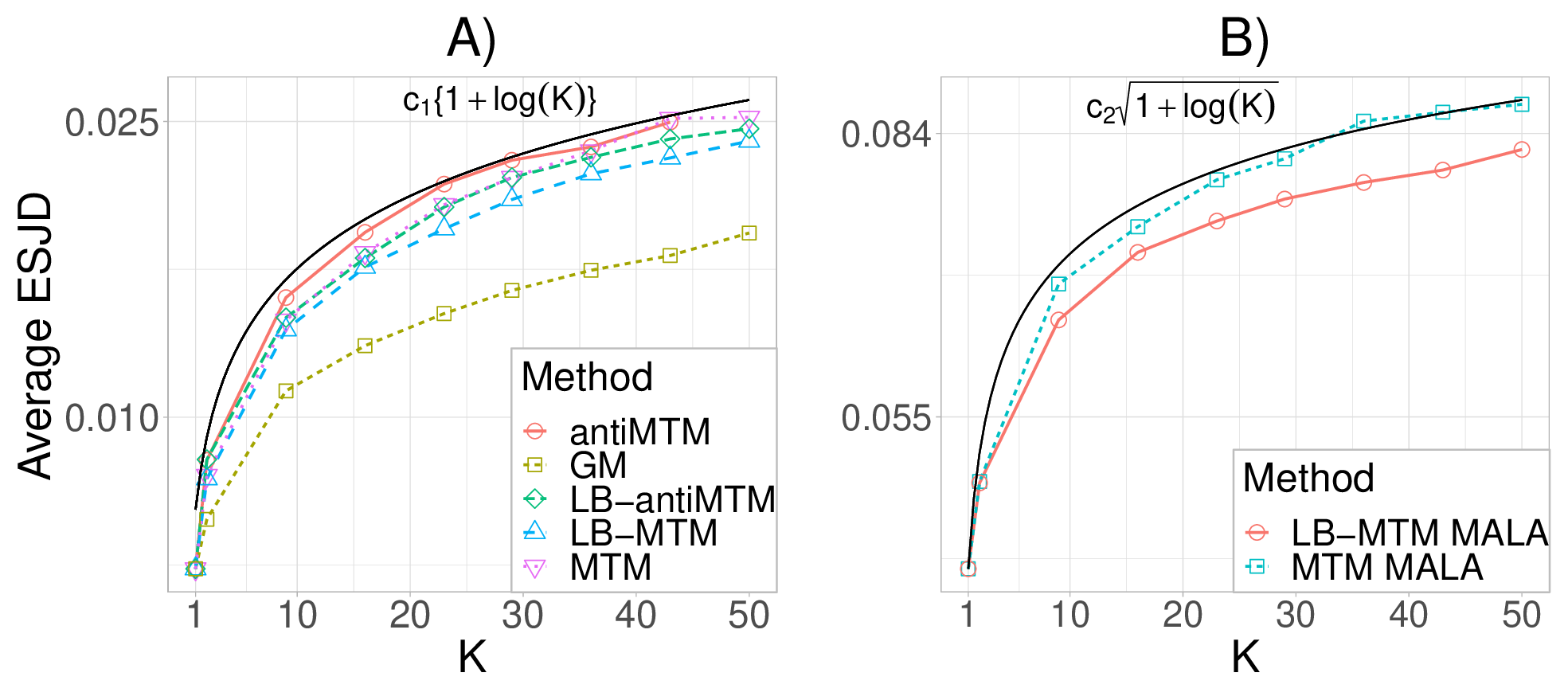}
    \caption{Estimated ESJD as a function of $K$, for MP-MCMC schemes with random walk (A) and Langevin (B) proposal, as described in Section \ref{sec:sim}.
    Black solid lines indicate the functions  $c_1\lbrace 1 + \log(K)\rbrace$ with  $c_1 = 2.3 E_{\textsc{rw}} $ for (A) and $c_2 \sqrt{1+\log(K)}$ with $ c_2 = E_{\textsc{mala}} $ for (B), where $E_{\textsc{rw}}$ and $E_{\textsc{mala}}$ denote the ESJD of MH with $K=1$ and, respectively, random walk and Langevin proposal.}
    \label{fig:log1}
\end{figure}

Figure \ref{fig:log1}-A) reports the estimated expected squared jump distance (ESJD), averaged over the $d$ components of $\btheta$.
For each scheme, the proposal step-size is optimised to maximise the estimated ESJD.
Taking the latter as an empirical measure of performance, all MTM schemes perform similarly. 
Instead,  GMH exhibits a slightly lower ESJD, although it should be noted that GMH requires roughly half target evaluations per iteration relative to MTM (see e.g.\ Section \ref{sec:par_comp}).
More importantly, we see that, regardless of the specific variant, all MP-MCMC schemes display a logarithmic growth in performances as $K$ increases, in accordance with Theorem \ref{thm:GRW}. The growth is actually dominated by the function $c_1 \lbrace 1 + \log(K)\rbrace$, suggesting that the upper bound of order $\lbrace 1 + \log(K)\rbrace^2$ is potentially conservative in this case. 

A similar behaviour is observed for MP-MCMC schemes with gradient-based proposals. In particular, Figure \ref{fig:log1}-B) displays the same quantities as Figure \ref{fig:log1}-A) for two MTM algorithms with proposals $Q(\btheta,d(\btheta'_1,\dots,\btheta'_K))=\prod_{i=1}^KQ_i(\btheta,d\btheta'_{i})$ and $Q_i(\btheta,\cdot)=N(\btheta+\sigma^2/2\nabla\log\pi(\btheta|\bz),\sigma^2\I_d)$.
In this setting, while the use gradient information leads to higher ESJD values, the growth with $K$ is even weaker, roughly following the function $c_2\sqrt{1+\log(K)}$.
Indeed, we expect theoretical results similar to Theorem \ref{thm:GRW} to hold also for gradient-based MP-MCMC schemes. 
However, while Theorem \ref{theo:1} applies to any proposal, the challenge to extend Theorem \ref{thm:GRW} to gradient-based schemes would be to control explicitly the term $\tilde{P}( x,\sX \backslash \{ x\})$ as a function of the proposal stepsize \citep{chen2023simple} and, more crucially, the need to have explicit and tight lower bounds on the spectral gap for $K=1$
in order to appropriately interpret the results.

% \subsection*{Acknowledgements} 
% Francesco Pozza and Giacomo Zanella are funded by the European Union (ERC), through the Starting grant `PrSc-HDBayLe', project number 101076564. 
% %Views and opinions expressed are however those of the author(s) only and do not necessarily reflect those of the European Union or the European Research Council. Neither the European Union nor the granting authority can be held responsible for them.

\bibliographystyle{biometrika}
\bibliography{paper-ref}

\appendix
\setcounter{equation}{0}
\renewcommand{\theequation}{\thesection.\arabic{equation}}
\setcounter{theorem}{0}
\renewcommand{\thetheorem}{\thesection.\arabic{theorem}}

\section{Concise overview of multiproposal MCMC methods}\label{sec:review_MP_MCMC}
This section provides a more detailed discussion and review of MP-MCMC methods, in particular Multiple-try Metropolis (MTM) and Generalised Metropolis-Hastings (GMH) ones mentioned in Section \ref{sec:review} of the paper.
While not technically novel, we also provide self-contained propositions (i.e.\ Propositions \ref{prop:kernel:mtm} and \ref{prop:kernel:gm})
showing the validity of those schemes and highlighting how they fit into the general framework of Algorithm \ref{alg:meta}. 

\subsection{Multiple-try Metropolis methods}\label{sec:MTM}
Although the term multiple-try Metropolis  encompasses various different algorithms, on a macroscopic level each method can be characterised by the specification of some basic components: a transition kernel $Q$ from $\sX$ to $\sX^K$ and $K$ non-negative weighting functions $w_i:\sX^2\to\R^+$. 
For example, when $\sX=\mathbb{R}^d$ and $\pi$ denotes the target density with respect to Lebesgue, a standard option is to take independent and identically distributed Gaussian increments, i.e.\ $Q( x,d\by_{1:K}) = \prod_i^K Q_i( x,d y_i)$ with $Q_i( x,\cdot) = N( x, \sigma^2 \I_d)$, $\sigma>0$, and weight function equal to $w_i( x, y) = \pi( y)$. 

Using these ingredients, multiple-try Metropolis algorithms proceed as described in Algorithm \ref{alg:MTM}, which can be obtained by combining Algorithm \ref{alg:meta} and Example \ref{ex:MTM} in the paper. 
Here, for every $i\in\{1,\dots,K\}$, $ x, y_i\in\sX$ and $\by_{-i}=( y_j)_{j\neq i}\in\sX^{K-1}$,  we define the normalised weights as
\begin{equation}\label{eq:norm_weights}
\bar{w}_i( x,  y_i;\by_{-i})=
\frac{w_i( x,  y_i)}{w_i( x,  y_i)+\sum_{j\neq i} w_j( x, y_j)}\,.
\end{equation}
\begin{algorithm}
\caption{Multiple-try Metropolis algorithm }\label{alg:MTM}
\KwIn{$ X_0 \in \sX$, a Markov transition kernel $Q$ from $\sX$ to $\sX^{K}$, a set of weight functions $w_i:\sX^2\to\R^+$ for $i=1,\dots,K$}
\For{$t = 0,1,\dots$ 
}{
Given $ X_t= x$, sample $\by_{1:K}\sim Q( x,\cdot)$
\\
Sample $i\in\{1,\dots,K\}$ with probabilities $\{\bar{w}_i( x,  y_i;\by_{-i})\}_{i=1}^K$\\
Sample $\by'_{-i} \sim \bar{Q}^{ x}_i( y_i, \cdot)$\\
Set $ X_{t+1} =  y_i$ with probability
$$\alpha_i( x,  y_i,\by_{-i},\by'_{-i})
=
\min \left( \, 1\, , \,  \frac{\bar{w}_i( y_i,  x;\by'_{-i})\pi(d y_i)Q_i( y_i,d x) 
}{
\bar{w}_i( x,  y_i;\by_{-i})\pi(d x)Q_i( x,d y_i) 
} \right), $$
and otherwise $ X_{t+1} =  x$
}
\end{algorithm} 

Algorithm \ref{alg:MTM} is a special case of Algorithm \ref{alg:meta} (or equivalently its transition kernel in a special case of \eqref{kernel:meta:1}) that is guaranteed to produce a $\pi$-reversible Markov chain, as shown in the following proposition.
\begin{proposition} \label{prop:kernel:mtm}
The Markov chain $\{ X_t\}$ induced by Algorithm \ref{alg:MTM} is $\pi$-reversible and its transition kernel coincides with \eqref{kernel:meta:1} by defining the function $h:\sX^{1+K}\to S_K$ as
\begin{equation}\label{eq:mtm_h}
h_i( x, \by_{1:K}) =   
\bar{w}_i( x,  y_i;\by_{-i})
\int_{\sX^{K-1} }  \alpha_i( x,  y_i,\by_{-i},\by'_{-i})\bar{Q}^{ x}_i( y_i, d\by'_{-i}),  
\end{equation}
for $i = 1,\dots,K$, $ x\in\sX$ and $\by_{1:K}\in\sX^K$.
\end{proposition}
\begin{proof}
Let $P$ be the transition kernel of Algorithm \ref{alg:MTM}, i.e.\ $P( x,A) = \Pr( X_{t+1}\in A\mid  X_t= x)$ for all $ x\in\sX$ and $A\subseteq \sX$ with $ X_t$ and $ X_{t+1}$ as in Algorithm \ref{alg:MTM}.
For every $A \subset \sX$ and $ x \notin A$, integrating over possible values of $\by_{1:K}$, $i$ and $\by'_{-i}$ in Algorithm \ref{alg:MTM}, we obtain
\begin{equation} \label{help:transition:mt1}
\begin{aligned}
& P( x,A) = \, \sum_{i = 1}^K \int_{\sX^{K}}  \mathbf{1}_{A}( y_i)\bar{w}_i( x,  y_i;\by_{-i})
\Big( 
\int_{\sX^{K-1}}  
\alpha_i( x,  y_i,\by_{-i},\by'_{-i})
\bar{Q}^{ x}_i( y_i, d\by'_{-i}) 
\Big)
Q( x, d\by_{1:K})\,.
\end{aligned}
\end{equation}
To prove that $P$ is $\pi$-reversible, 
take $A,B \subseteq X$ with $A \cap B =  \varnothing$.
Then, by \eqref{help:transition:mt1} and using the decomposition 
$Q( x, d\by_{1:K})=Q_i( x,d y_i)\bar{Q}^{ y_i}_i( x, d\by_{-i})$, we have
\begin{align}
&\int_{\sX} \mathbf{1}_{B}( x) P( x,A)\pi(d x)  
\nonumber \\
&= \sum_{i = 1}^K
\int_{\sX^{2K}}  
\mathbf{1}_{B}( x) \mathbf{1}_{A}( y_i)
\bar{w}_i( x,  y_i;\by_{-i})    
\alpha_i( x,  y_i,\by_{-i},\by'_{-i})
\bar{Q}^{ x}_i( y_i, d\by'_{-i}) 
\bar{Q}^{ y_i}_i( x, d\by_{-i})Q_i( x, d y_i)\pi(d x)
\nonumber \\
&=  \sum_{i = 1}^K \int_{\sX^{2 K}}  \mathbf{1}_{B}( x) \mathbf{1}_{A}( y_i)   
\mu(\by_{-i},\by'_{-i};d x,d y_i)
\bar{Q}^{ y_i}_i( x, d\by_{-i})\bar{Q}^{ x}_i( y_i, d\by'_{-i})\label{eq:rev_MTM_1}\\
&= \int_{\sX}  \mathbf{1}_{A}( y_i) \pi(d y) P( y,B),\label{eq:rev_MTM_2}
\end{align}
where for every  $\by_{-i},\by'_{-i}\in\sX^{K-1}$, $\mu(\by_{-i},\by'_{-i};\cdot,\cdot)$ is a measure on $\sX^2$ defined as
\begin{align*}
\mu(\by_{-i},\by'_{-i};d x,d y_i)
=&
\alpha_i( x,  y_i,\by_{-i},\by'_{-i})\bar{Q}^{ x}_i( y_i, d\by'_{-i}) \pi(d x)\\
=&
\min\left\{
\pi(d x)Q_i( x,d y_i)\bar{w}_i( x,  y_i;\by_{-i})
\;,\;
\pi(d y_i)Q_i( y_i,d x)\bar{w}_i( y_i,  x;\by'_{-i})
\right\},
\end{align*}
and  the equality in \eqref{eq:rev_MTM_2} results from the symmetry with respect to $( x,\by'_{-i})$ and $( y_i,\by_{-i})$ of the expression in \eqref{eq:rev_MTM_1}.
The minimum of two measures in the definition of $\mu$ indicates the measure having density equal to the minimum of the two densities (with respect to an arbitrary common dominating measure).
Finally, note that $\sum_{i = 1}^K h_i( x, \by_{1:K}) \leq 1 $ because
\begin{equation*}
\sum_{i = 1}^K h_i( x, \by_{1:K}) =   \sum_{i = 1}^K\bar{w}_i( x,  y_i;\by_{-i}) \int_{\sX^{K-1}}   \alpha_i( x,  y_i,\by_{-i},\by'_{-i}) \bar{Q}^{ x}_i( y_i, d\by'_{-i})
\leq\sum_{i = 1}^K\bar{w}_i( x,  y_i;\by_{-i})=1
,
\end{equation*}
since $\alpha_i\leq 1$.
\end{proof}

Over the last two decades, many MTM variants have been proposed, which differ in the choice of kernel $Q$ and weight functions $w_i$. 
We briefly review some of those below.
The original formulation of \citet{liu2000multiple} requires the proposal density to be the product of $K$ independent components, i.e.\ $Q(x, d\by_{1:K}) = \prod_{i = 1}^K Q_i(x,d y_i)$. Practical choices for the weight functions therein include $w(x,y) = \pi(y)$, where $\pi$ denotes the target density. 
\citet{craiu2007acceleration} suggest to use negatively correlated proposals distributions to explore the sample space more efficiently. This is done by considering Gaussian random variables that, given $ X_t =  x$, are marginally distributed as $N( x, \sigma^2 \I_d)$ and for which the average distance between two proposed states $ Y_i,  Y_j,$  $i,j = 1,\dots,K$ is maximal or, equivalently, the pairwise correlation between each couple $ Y_i,  Y_j,$ $i \neq j$ is minimal. Such a result is obtained by considering proposals of the form $( Y_1,\dots, Y_K) \sim N( \boldsymbol{\mu}_{\textsc{anti}}, \boldsymbol{\Sigma}_{\textsc{anti}} ) $ where 
$\boldsymbol{\mu}_{\textsc{anti}} = ( x,\dots, x) \in \R^{dK}$ and $\boldsymbol{\Sigma}_{\textsc{anti}} \in \R^{dK\times dK} $ is positive definite matrix defined as
\begin{equation*}
\boldsymbol{\Sigma}_{\textsc{anti}} = 
\begin{bmatrix}
\sigma^2 \I_d & \boldsymbol{\Psi} & \boldsymbol{\Psi}& \cdots & \boldsymbol{\Psi}     \\
\boldsymbol{\Psi} &\sigma^2 \I_d & \boldsymbol{\Psi}& \cdots & \boldsymbol{\Psi} \\
\vdots & \ddots &\ddots & \ddots &\vdots \\
\boldsymbol{\Psi} & \cdots & \cdots& \cdots &\sigma^2 \I_d \\
\end{bmatrix}  ,
\end{equation*}
with $\boldsymbol{\Psi} = -(\sigma^2/(K-1))\I_d$. Making use of quasi Monte Carlo techniques, \citet{bedard2012scaling} develop a multiple-try version of the classical hit-and-run algorithm which removes the need to sample the $K-1$ additional shadow points to maintain the chain $\pi$-invariant. \cite{casarin2013interacting} develop a multiple-try metropolis algorithm in which the $K$ proposal distributions are independent but not identically distributed. 
A general formulation of MTM schemes similar to the one we provide in Algorithm \ref{alg:MTM} is given in \citet{pandolfi2010generalization}.
\citet{gagnon2023improving} identify the use of globally-balanced weighting functions as the reason for some pathological behaviors of multiple-try metropolis observed by \citet{martino2017issues}. As a solution to this issue, the authors propose the adoption of weighting function which are locally balanced \citep{zanella20}, i.e., satisfy $w( x, y)=g(\pi( y)/\pi( x))$ or $w( x, y)=g(\pi( y)q(y,x)/\pi( x)q(x,y))$ with $g\, : \, \R \to \R^{+}$ such that $g(t)/g(1/t) = t$ for every $t>0$, where $\pi$ and $q$ denote densities of, respectively, $\pi(d x)$ and $Q( x,d y)$ with respect to a common dominating measure. 

\subsection{Generalised Metropolis-Hastings methods} \label{section:GMHA}
Generalised Metropolis-Hastings (GMH) \citep{tjelmeland2004using,calderhead2014general,holbrook2023generating}
are a class of MP-MCMC schemes where the selection probabilities $h$ in Algorithm \ref{alg:meta} are chosen so that it holds
\begin{align}\label{eq:rev_GMH}
h_i( x,\by_{1:K}) =& r_i( x,\by_{1:K})
h_i( y_i,( x,\by_{-i})),& x\in\sX,\,\by_{1:K}\in\sX^K\,,i\in\{1,\dots,K\},
\end{align}
where $r_i:\sX\times \sX^{K}\to[0,\infty)$ denotes the Radon-Nikodym derivative
\begin{equation}\label{eq:MH_ratio}
r_i( x,\by_{1:K})=\frac{\pi(d y_i) Q( y_i,d( x,\by_{-i}))}{\pi(d x) Q( x,d\by_{1:K} )}\,,
\end{equation}
and, with a slight abuse of notation, we denote
$$
( x,\by_{-i})=( y_1,\dots, y_{i-1}, x, y_{i+1},\dots, y_{K})\in\sX^K.
$$
Condition \eqref{eq:rev_GMH} is a sufficient (though not necessary) condition for the resulting MP-MCMC scheme to be $\pi$-reversible, as shown by the following proposition.
\begin{proposition} \label{prop:kernel:gm}
If \eqref{eq:rev_GMH} holds, the kernel $P^{(K)}$ defined in \eqref{kernel:meta:1} is $\pi$-reversible.
\end{proposition}
\begin{proof}
By \eqref{kernel:meta:1}, for every $A,B \subseteq X$ with $A \cap B =  \varnothing$, we have
\begin{align*}
\int_{B}P^{(K)}( x, A)\pi(d x)
= &
\sum_{i= 1}^{K}
\int_{\sX^{K+1} } 
\mathbf{1}_{B}( x)
\mathbf{1}_{A}( y_i)
h_i( x,\by_{1:K})
\pi(d x) Q( x, d\by_{1:K})
\\= &  
\sum_{i= 1}^{K}
\int_{\sX^{K+1} } 
\mathbf{1}_{B}( x)
\mathbf{1}_{A}( y_i)
h_i( y_i,( x,\by_{-i}))
r_i( x,\by_{1:K})
\pi(d x) Q( x, d\by_{1:K})
\\= &
\sum_{i= 1}^{K}
\int_{\sX^{K+1} } 
\mathbf{1}_{B}( x)
\mathbf{1}_{A}( y_i)
h_i( y_i,( x,\by_{-i}))
\pi(d y_i) Q( y_i, d( x,\by_{-i}))
\\= &
\int_{A}P^{(K)}( y_i, B)\pi(d y_i),
\end{align*}
where the second equality follows from \eqref{eq:rev_GMH}, the third by definition of Radon-Nykodim derivative and the fourth again by \eqref{kernel:meta:1}.
\end{proof}
\begin{remark}[Data augmentation construction]
GMH schemes are usually presented and derived starting from a data augmentation approach \citep{tjelmeland2004using}, which works with Markov chains defined on an augmented state space $\sX^{K+1} \times \{0,1,\dots,K\}$. 
In this construction, GMH schemes can be interpreted as Markov chains alternating specific conditional updates in the augmented space, see Section 2 of \cite{tjelmeland2004using} for details. 
In our exposition, we focused on the single chain interpretation and construction to highlight how GMH schemes fit into the framework of Algorithm \ref{alg:meta}. 
In most situation, the two constructions end up being equivalent even if the data augmentation one allows in principle for more flexibility, see e.g.\ \citet{holbrook2023generating} for more discussion.
\end{remark}

\begin{remark}[Condition \eqref{prop:kernel:gm} for singular measures]
Similarly to what is done for the classical MH algorithm, Proposition \ref{prop:kernel:gm} applies to general measures $\pi$ and kernel $Q$, by requiring condition \eqref{eq:rev_GMH} to hold on the set $R\subseteq \sX^{1+K}$ where the two measures in \eqref{eq:MH_ratio} are mutually absolutely continuous and 
by setting $h_i$ to be $0$ on the complement of $R$; see Section 2 of \cite{tierney1998note} for details in the classical MH case.
\end{remark}
\citet{tjelmeland2004using} discusses various choices of functions satisfying \eqref{eq:rev_GMH}, such as the one described in Example \ref{ex:GMH} of the paper, which is denoted as (\textit{T1}) in \citet{tjelmeland2004using}.
The latter is not the optimal choice in general, even if the suboptimality factor is usually small and decreases with $K$.
See \citet{tjelmeland2004using,calderhead2014general, holbrook2023generating} for more details. As mentioned in Example \ref{ex:GMH} of the paper, good choices of $Q$ are ones where the density $q$ of $Q$, with respect to some dominating measure $\mu^K$ on $\sX^K$,
satisfies 
\begin{equation}\label{eq:exch}
q( y_i,d( x,\by_{-i}))=q( x,\by_{1:K})\,,    
\end{equation}
so that \eqref{acc:pr:tj} simplifies to
\begin{equation*}%\label{acc:pr:tj}
h_{i}( x,\by_{1:K}) = 
\frac{
\pi( y_i) 
}{
\pi( x)
+
\sum_{j =1}^{K} 
\pi( y_j)
}\,,
\end{equation*}
where $\pi$ denotes the density of $\pi$ with respect to $\mu$.
Intuitively, \eqref{eq:exch} ensures that, when $\by_{1:K}$ is sampled from $Q( x,\cdot)$, then $( x,\by_{-i})$ is also a plausible sample from $Q( y_i,\cdot)$ for $i=1,\dots,K$, which is required for $h_{i}( x,\by_{1:K})$ not to be too small.
Examples are proposals satisfying \eqref{eq:exch} include the proposal (P1) of \citet{tjelmeland2004using}, which is described in Example \ref{ex:GMH} and it involves the simulation of a latent variable $z\sim N( x, (\sigma^2/2)\I_d)$,
and the simplicial sampler of \citet{holbrook2023generating}. 
The latter generates $K$ candidate points by randomly rotating a $K$-dimensional simplex around the current state of the chain. Given $K$ vertices $v_1, \dots, v_{K},\in \R^d$ satisfying $\|v_i\| = \lambda$, $\|v_i - v_j \| = \lambda$, for $i\neq j$ and $\lambda>0$, a random rotation $\mathbf{Q}$ is first generated $\mathbf{Q} \sim H$, where $H$ denotes the Haar measure on $\{\mathbf{A}\in \R^{d \times d} \, : \,\mathbf{A}^{\top}\mathbf{A} = \I_d \}$. Then, the proposed states are obtained as $ y_i =  x + \mathbf{Q}v_i$, $i = 1,\dots,K$. See also \citet{tjelmeland2004using,calderhead2014general,luo2019multiple,holbrook2023generating} for further options concerning the specification of both $Q$ and $h$. 

\subsection{General remarks}
To the best of our knowledge, there is no strict ordering of performances between MTM and GMH schemes in general.
In particular, GMH schemes do not require the generation of additional shadow points and therefore, when considering the same number of iterations, they tend to have a computational cost that is approximately half of that of MTM schemes \citep{holbrook2023generating}. 
On the other hand, in order to satisfy \eqref{eq:exch}, common GMH proposals end up proposing $K$ positively correlated points, which tends to reduce sampling efficiency (relative to MTM which can use, for example, antithetic proposals).
This in principle results in a trade-off between computational cost and sampling efficiency between MTM and GMH, so that the resulting best performing scheme is in general case dependent.

Crucially, all the schemes discussed above, as well various others proposed in the literature, fall into the general framework of Algorithm \ref{alg:meta}.
Thus, the results of Section \ref{sec:teo:results} applies to all of them, suggesting that they all share fundamental limitations in terms of potential improvements with respect to standard MH as $K$ increases.

\begin{remark} 
There are other ways to parallelise MCMC schemes proposed in the literature that, while being similar to multiproposal schemes, do not technically fit into the framework of Algorithm \ref{alg:meta}.
Examples include prefetching methods \citep{brockwell2006parallel} or the multi-core MH schemes discussed in \citet[Ch.4]{power2020}.
At a high level, the difference with Algorithm \ref{alg:meta} is that, instead of using $K$ parallel evaluations to simulate a single step of a faster Markov chain, these schemes use parallel evaluations to simulate multiple steps of the original MH chain.
Thus, since they do not fit (at least not in a straightforward way) into the framework of Algorithm \ref{alg:meta}, the results we develop in this work do not directly apply to them. 
Nonetheless, we note that a poly-logarithmic gain in efficiency with respect to $K$, like the one we prove for MP-MCMC methods in Section \ref{sec:neg3}, has also been derived with an optimal scaling analysis for some of those methods in \citet[Ch.4]{power2020}.
\end{remark}

\section{Proofs and additional technical results}
\begin{proposition}\label{prop:equi:kern:MA}
Algorithm \ref{alg:meta} produces a Markov chain with transition kernel as in \eqref{kernel:meta:1}. 
\end{proposition}
\begin{proof}
Let $P$ be the transition kernel of Algorithm \ref{alg:meta}, i.e.\ for all $ x\in\sX$ and $A\subseteq \sX$ define
\begin{align*}
P( x,A) = &\Pr( X_{t+1}\in A\mid  X_t= x)\,,
\end{align*}
with $ X_t$ and $ X_{t+1}$ as in Algorithm \ref{alg:meta}.
Denoting by $\bY_{1:K}\sim Q( x,\cdot)$ the random variables generated in the first step of Algorithm \ref{alg:meta} (using $\by_{1:K}$ only for their realizations for clarity) and 
by $I$ the index of the state selected in the second step of Algorithm \ref{alg:meta}, with the convention $I = 0$ if $ X_{t+1} =  x$, for every $A \subset \sX$ and $ x \in \sX\backslash A$, we have
\begin{align} 
P( x,A) = & \sum_{i = 1}^K 
\Pr( X_{t+1}\in A,I=i\mid  X_t= x) 
\nonumber\\=& 
\sum_{i = 1}^K \int_{\sX^K} \mathbf{1}_{A}( y_i) \Pr( I = i\mid  X_t= x, \bY_{1:K}=\by_{1:K}) Q( x, d\by_{1:K})
\nonumber\\
=& 
\sum_{i = 1}^K \int_{\sX^K} \mathbf{1}_{A}( y_i) h_i( x,\by_{1:K}) Q( x, d\by_{1:K})\,,\label{help:equiv:rep:1}
\end{align}
which implies that $P=P^{(K)}$ with $P^{(K)}$ as in \eqref{kernel:meta:1}. Note that \eqref{help:equiv:rep:1} defines the value of $P( x,A)$ even if $ x \in A $ since $P( x,A) = 1 - P( x,\sX\backslash A)$  with $ P( x,\sX\backslash A)$ defined as in \eqref{help:equiv:rep:1}.
\end{proof}
\subsection{Proof of Theorem \ref{theorem:peskun}}
\begin{proof}
Let $A \subset \sX$ and $ x \notin A$.
By \eqref{kernel:meta:1} and $h_i(  x, \by_{1:K})\leq 1$, we have 
\begin{equation}\label{eq:ineq_1}
P^{(K)}( x,A) \leq \sum_{i = 1}^K \int_{\sX^{K}}  \mathbf{1}_{A}( y_i)  Q(  x, d\by_{1:K}) =  \sum_{i = 1}^K Q_i( x,A).
\end{equation}
Using first $\pi$-reversibility of $P^{(K)}$ and then \eqref{eq:ineq_1}, we have
\begin{align}
\mathbf{1}( x\neq y)
\pi(d x)P^{(K)}( x,d y)
= & 
\mathbf{1}( x\neq y)
\min \left(  
\pi(d x)P^{(K)}( x,d y)
,  
\pi(d y)P^{(K)}( y,d x)
\right) 
\nonumber\\
\leq & 
\mathbf{1}( x\neq y)
\min \left(  \pi(d x)\sum_{i=1}^K Q_i(  x,d y),  \pi(d y)\sum_{i=1}^KQ_i( y,d x)\right) 
\label{eq:min_two_measures}\,,
\end{align}
where the above inequalities apply to measures defined on the product space $\sX\times\sX$ (as in e.g.\ \citealp[Sec.2]{tierney1998note}), the minimum of two measures indicates the measure having density equal to the minimum of the two densities (with respect to an arbitrary common dominating measure), and $\mathbf{1}( x\neq y)$ denotes the function equal to $1$ if $ x\neq  y$ and zero otherwise.
Then, using $\tilde{Q}=K^{-1}\sum_{i=1}^KQ_i$ and the definition of MH kernel, the right-hand side of \eqref{eq:min_two_measures} is equal to 
\begin{align}
\mathbf{1}( x\neq y)
K \min \left(  \pi(d x)\tilde{Q}(  x,d y),  \pi(d y)\tilde{Q}( y,d x)\right) 
\nonumber= & 
\mathbf{1}( x\neq y)
K \pi(d x)\tilde{Q}( x,d y) \min \left(  1, \frac{ \pi(d y)\tilde{Q}( y,d x)}{\pi(d x)\tilde{Q}( x,d y)} \right)
\nonumber\\=&
\mathbf{1}( x\neq y)
K \pi(d x)\tilde{P}( x,d y)
\,.
\label{eq:min_two_measures_2}
\end{align}
Combining \eqref{eq:min_two_measures} and \eqref{eq:min_two_measures_2} gives \eqref{mt:marginal:upper:bound:mh}.
\end{proof}

\subsection{Proof of Corollary \ref{coroll:peskun_gap}}
\begin{proof}
The statement follows directly from Theorem \ref{theorem:peskun} and classical results about the implications of Peskun ordering, e.g.\ \citet[Lemma 32]{andrieu2018uniform}.
\end{proof}

\subsection{Proof of Theorem \ref{theo:1}} \label{proof:theo:1}
The following lemma is a classical bound that holds for any $\pi$-reversible chain, of which we report a self-contained proof for the readers convenience.
\begin{lemma}\label{lemma:conductance_bound}
Let $P$ be a $\pi$-reversible kernel. Then
\begin{equation} \label{eq:cond_bound}
Gap(P) \leq 2 \essinf_{ x\in\sX} P( x,\sX \backslash \{ x\}),
\end{equation}
where $\essinf$ refers to the essential infimum with respect to $\pi$.
\end{lemma}
\begin{proof}[ of Lemma \ref{lemma:conductance_bound}]
We start from the bound 
\begin{equation} \label{help:thm:0}
Gap(P) \leq 2 \Phi(P),
\end{equation}
where
\begin{equation} \label{def:conductance}
\Phi(P) = \inf\Biggl\{ \frac{\int_{A}P( x, A^c)\pi(d x)}{\pi(A)}\, : \, A \subset \sX,  0 < \pi(A)
\leq 0.5 \Biggr\} , 
\end{equation}
is the conductance of $P$ (see e.g. Lemma 5 of \citealp{andrieu2022explicit}).
The numerator in \eqref{def:conductance} can be upper bounded as
\begin{equation*} 
\begin{aligned}
\int_{A} P( x, A^c) \pi(d x) 
\leq & 
\int_{A} P( x, \sX\backslash\{ x\})
\pi(d x)
\leq 
\pi(A) \esssup_{ x \in A} P( x, \sX\backslash\{ x\})\,.
\end{aligned}
\end{equation*}
so that, by \eqref{help:thm:0}, it follows
\begin{equation} \label{help:thm:1}
Gap(P) \leq 2 \inf \Biggl\{   \esssup_{ x \in A} P( x, \sX\backslash\{ x\})  \, \, : \, A \subset \sX,  0 < \pi(A)
\leq 0.5  \Biggr\}.    
\end{equation}
Note now that, for any $\delta >0$, there exist $A$ such that $0<\pi(A) \leq 0.5$ and $\esssup_{ x \in A} P( x, \sX\backslash\{ x\})   < \essinf_{ x \in \sX} P( x, \sX\backslash\{ x\})   + \delta$ and thus, by \eqref{help:thm:1}, 
\begin{equation*}
Gap(P) \leq 2 \essinf_{ x \in \sX}P( x, \sX\backslash\{ x\}) + \delta.  
\end{equation*}
By the arbitrariness of $\delta$ we obtain \eqref{eq:cond_bound}.
\end{proof}
\begin{proof}[ of Theorem \ref{theo:1}]
Taking $A=\sX\backslash\{ x\}$ in Theorem \ref{theorem:peskun} gives 
$P^{(K)}( x, \sX\backslash\{ x\})\leq K \tilde{P}( x, \sX\backslash\{ x\})$ for all $ x\in\sX$ which, combined with Lemma \ref{lemma:conductance_bound}, implies
\begin{equation*}
Gap(P^{(K)}) \leq 2 \essinf_{ x\in\sX} P^{(K)}( x,\sX \backslash \{ x\})
\leq  2 K \essinf_{ x\in\sX} \tilde{P}( x,\sX \backslash \{ x\})\,.
\end{equation*}

The second term in the minimum function of \eqref{theo1:result1} can be obtained by considering the class of linear functions in the definition of spectral gap, as follows.
Let $ X_t\sim \pi$ and $ X_{t+1}| X_t\sim P^{(K)}( X_t,\cdot)$.
By definition of Algorithm \ref{alg:meta}, we have
$ X_{t+1}= Y_I$ with $( Y_1,\dots, Y_K)\sim Q( X_t,\cdot)$ being the $K$ candidate states proposed in the first step of the algorithm, and $I$ being the random index selected from $\{0,1,\dots,K\}$ in the second step of the algorithm, where $ Y_0= X_t$ by convention. 
It follows that $(\nu^{\top}( X_{t+1} - X_t))^2
\leq 
\max_{i = 1,\dots,K} (\nu^{\top}( Y_i - X_t))^2$.
Thus, from the definition of spectral gap we have
\begin{equation*}
\begin{aligned}
Gap(P^{(K)}) \leq &
\inf_{\nu \in \mathbb{R}^d,\,\|\nu\|=1} 
\frac{
E( (\nu^{\top}( X_{t+1} - X_t))^2)
}{
2 \mathrm{Var}(\nu^{\top}  X_t )
}
\leq
\inf_{\nu \in \mathbb{R}^d,\,\|\nu\|=1} 
\Biggl(\frac{E (\max_{i = 1,\dots,K} (\nu^{\top}( Y_i - X_t))^2)}{2 \mathrm{Var}(\nu^{\top}  X_t)} \Biggr)\,,
\end{aligned}
\end{equation*}
 as desired.
 Note that the above infimum should be intended over vectors $\nu$ such that 
 $\mathrm{Var}(\nu^{\top}  X_t)$ is in $(0,\infty)$.
\end{proof}

\subsection{Proof of Corollary \ref{coroll:mgf}} 
To prove the result, we use a standard trick for upper bounding the expectation of the maximum of positive random variables with finite moment generating function  \citep[see e.g.,][]{dasarathy2011simple}. 
Denoting
$A_{\nu,i}^2=(\nu^{\top}( Y_i - X_t))^2$ for brevity, 
it follows from Jensen inequality and $E \big( \exp\big( s A_{\nu,i}^2 )\big) \leq M^{\nu}(s)$ 
that 
\begin{equation} \label{help1:ub:nu}
\begin{aligned}
\exp\big( s E\big(\max_{i = 1,\dots,K} A_{\nu,i}^2 \big) \big) \leq &\, E \big( \max_{i= 1,\dots,K} \exp\big( s A_{\nu,i}^2 )\big)
\leq \sum_{i = 1}^K E \big( \exp\big( s A_{\nu,i}^2 )\big) \, \leq \, K M^{\nu}(s).
\end{aligned}
\end{equation}
By taking the logarithm of both sides of \eqref{help1:ub:nu} it follows
\begin{equation*}
E\Big(\max_{i = 1,\dots,K} A_{\nu,i}^2 \Big) \leq \frac{\log(K) + \log(M^{\nu}(s))}{s} .
\end{equation*}
The inequality in \eqref{theo1:result2} follows by combining the above with Theorem \ref{theo:1}.

\subsection{Proof of Theorem \ref{thm:GRW}}
\begin{proof}
Assumption \ref{cond:2} implies that $\tilde Q( x,\cdot)=K^{-1}\sum_{i=1}^KQ_i( x,\cdot)=N( x,\sigma^2\I_d)$.
Thus
\begin{equation} \label{help:gauss:acc:0}
\begin{aligned}
\tilde{P}( x,\sX \backslash \{  x \}) 
=& \int \min \Big \lbrace 1 , \frac{\pi( x + \sigma \varepsilon_i)}{ \pi( x)} \Big \rbrace \phi_d(\varepsilon; \mathbf{0}, \I_d) d\varepsilon,
\end{aligned}
\end{equation}
 with $\phi_d(\cdot; \mathbf{0}, \I_d)$ denoting the multivariate Gaussian probability density function with mean $\mathbf{0}$ and covariance matrix $\I_d$.
This is exactly the acceptance probability of standard MH with random walk Gaussian proposals, which under Assumption \ref{cond:1} can be bounded as 
\begin{equation*}
\essinf_{ x \in \sX} \tilde{P}( x, \sX\backslash\{ x\}) \leq (1 + m \sigma^2)^{-d/2},
\end{equation*}
see e.g.\ the proof of Proposition 45 of \cite{andrieu2022explicit}. By \eqref{theo1:result2}, this implies
\begin{equation} \label{GRW:help:1stpart1}
\mathrm{Gap}(P^{(K)}) \leq 2K(1 + m \sigma^2)^{-d/2}.
\end{equation}
We now study the quantity 
\begin{equation*}
\inf_{\nu \in \mathbb{R}^d} \frac{1}{ 2 s} \frac{ \log(K) + \log(M^{\nu}(s))}{ \mathrm{Var}(\nu^{\top}  X)}.
\end{equation*}
with $ X\sim \pi$.
To this end first note that, for every $\nu \in \mathbb{R}^d$, from Assumption \ref{cond:1} and \citet[Equation 10.25]{saumard2014log} it follows that $\mathrm{Var}(\nu^{\top}  X) \geq \|\nu\|^2/L= 1/L$. 
In addition, the random variable $A_{\nu,i} = \nu^{\top}( Y_i -  x)$ is $N(0, \sigma^2)$-distributed implying, in turn, $\sigma^{-2}A_{\nu,i}^2\sim   \chi^2_1$.
As a result, the moment generating function of $A_{\nu,i}^2$ is equal to
\begin{equation*}
 E\Big( \exp(s A_{\nu,i}^2) \Big) =   
\big( 1 - 2 \sigma^2 s \big)^{-1/2},
\end{equation*}
for $s<1/(2 \sigma^2)$. 
Choosing $s = (1-\exp(-1))/(2 \sigma^2)$ gives $\log M^{\nu}(s) = 0.5$ and, as a direct consequence, 
\begin{equation} \label{GRW:help:1stpart2}
\mathrm{Gap}(P) \leq \frac{\sigma^2 L (\log(K) + 0.5)}{1-\exp(-1)}.
\end{equation}
The combination of \eqref{GRW:help:1stpart1}, \eqref{GRW:help:1stpart2} and $1-\exp(-1)>0.5$ proves \eqref{thm:GRW:result}.

We now prove \eqref{thm:GRW:result2}.
First note that, if $\log(K)\geq d$, then the right-hand side in \eqref{thm:GRW:result2} is larger or equal than $2$ and thus the inequality holds trivially.
Thus, we can assume $\log(K)< d$ without loss of generality.
Using \eqref{thm:GRW:result}, the change of variables $r=m\sigma^2$, and $L/m\geq 1$, we have
\begin{align*}
 \sup_{\sigma\in\R^+}Gap(P^{(K)})  \leq &\,
2\sup_{\sigma\in\R^+}  
\min  \Biggl( \frac{K}{(1 + m \sigma^2)^{d/2}},  \sigma^2 L \left(\frac{1}{2}+\log(K)\right)\Biggr)
\\= &\,
2\sup_{r\in\R^+}   
\min  \left( \frac{K}{(1 + r)^{d/2}},  \frac{L}{m} r \left(\frac{1}{2}+\log(K)\right)\right)
\leq
2\frac{L}{m}\sup_{r\in\R^+}   
\min  \left(f_1(r),f_2(r)\right),
\end{align*} 
for $f_1(r)=\exp(b-a\log(1+r))$ and $f_2(r)=r\left(1/2+b\right)$, with 
$a=d/2$ and $b=\log(K)$.
We will bound the above term using Lemma \ref{lemma:two_functions} below, whose assumptions are satisfied because $f_1$ and $f_2$ are, respectively, strictly decreasing and increasing in $r$ and $f_2(r)/f_1(r)\to\infty$ as $r\to\infty$.
Taking $r_0=\exp((b+\log(a))/a)-1$ and using the inequality $\exp(x)-1\geq x$, we have
$$
f_2(r_0)
=
r_0(1/2+b)
\geq
\frac{(b+\log(a))(1/2+b)}{a}
\geq \frac{1}{a}
=f_1(r_0)\,,
$$
where we also used $(b+\log(a))(1/2+b)\geq 1$, which follows from $d,K> 2$.
We can thus apply Lemma \ref{lemma:two_functions} to deduce
$\sup_{r\in\R^+}   
\min  \left(f_1(r),f_2(r)\right)\leq f_2(r_0)$ and thus
\begin{align*}
\sup_{\sigma\in\R^+} Gap(P^{(K)})
&\leq
\frac{2L}{m}r_0\left(\frac{1}{2}+b\right)
=
\frac{2L}{m}\left(\exp\left(\frac{2(\log(K)+\log(d/2))}{d}\right)-1\right)\left(\frac{1}{2}+\log(K)\right)\,.
\end{align*} 
Then, since $\log(K)< d$ and $d>2$, it is easy to show that  $2(\log(K)+\log(d/2))/d\leq 1.04$ and thus using $\exp(x)-1\leq \exp(c) x$ for $x\in(0,c)$, we obtain
\begin{align*}
\exp\left(\frac{2(\log(K)+\log(d/2))}{d}\right)-1
&\leq
\exp(1.04)
\frac{2(\log(K)+\log(d/2))}{d}\,.
\end{align*} 
Combining the above, we obtain
\begin{align*}
\sup_{\sigma\in\R^+} Gap(P^{(K)})
&\leq
c
\frac{L}{m}
\frac{\log(K)+\log(d/2)}{d}\left(\frac{1}{2}+\log(K)\right)
\leq
c 
\frac{L}{m}
\frac{(\log(K)+\log(d))^2}{d}\,,
\end{align*} 
with $c=4\exp(1.04)$, where we used $\log(d)\geq 1$ for $d>2$.
\end{proof}
\begin{lemma}\label{lemma:two_functions}
Let $f_1,f_2:\R^+\to\R^+$ continuous functions $f_1$ strictly decreasing, $f_2$ strictly increasing, and $\lim\sup_{r\to\infty}f_2(r)/f_1(r)>1$. Then
$$
\sup_{r\in\R^+}\min\{f_1(r),f_2(r)\}\leq f_2(r_0),
$$
for any $r_0$ such that $f_1(r_0)\leq f_2(r_0)$.
\end{lemma}
\begin{proof}
The statement follows from 
$
\sup_{r\in\R^+}\min\{f_1(r),f_2(r)\}= f_1(r^*)=f_2(r^*)\leq f_2(r_0)
$ where $r^*$ is the unique value in $\R^+$ such that $f_1(r^*)= f_2(r^*)$.
\end{proof}

\end{document}